\documentclass[runningheads,11pt]{llncs}
\usepackage{amsmath,amssymb,amsfonts,mathrsfs}
\usepackage[dvips,vmargin={1in,1in},hmargin={1in,1in}]{geometry}

\newcommand{\eqdef}{\stackrel{\text{def}}{=}}

\newcommand{\F}{\ensuremath{\mathbb{F}}}

\newcommand{\card}[1]{\left | #1 \right |}

\newcommand{\prob}{\ensuremath{\textsf{Prob}}}

\newcommand{\code}[1]{\ensuremath{\mathscr{#1}}}

\renewcommand{\dim}{\ensuremath{\textsf{dim}}}

\newcommand{\word}[1]{\ensuremath{\boldsymbol{#1}}}
\newcommand{\av}{\word{a}}
\newcommand{\bv}{\word{b}}

\newcommand{\cv}{\word{c}}

\newcommand{\ev}{\word{e}}

\newcommand{\uv}{\word{u}}
\newcommand{\vv}{\word{v}}

\newcommand{\xv}{\word{x}}
\newcommand{\yv}{\word{y}}

\newcommand{\Gv}{\word{G}}

\newcommand{\Pv}{\word{P}}

\newcommand{\Sv}{\word{S}}

\newcommand{\Xv}{\word{X}}
\newcommand{\Yv}{\word{Y}}
\newcommand{\Zv}{\word{Z}}

\newcommand{\mat}[1]{\ensuremath{\boldsymbol{#1}}}
\newcommand{\Gp}{\mat{G}}

\newcommand{\Gm}{\mat{G}}
\newcommand{\Hm}{\mat{H}}

\newcommand{\Pm}{\mat{P}}

\newcommand{\Sm}{\mat{S}}

\newcommand{\fq}{GF(q)}

\newcommand{\GRS}[3]{\text{\bf GRS}_{#1}(#2,#3)}
\newcommand{\RS}[2]{\text{\bf RS}_{#1}(#2)}

\spnewtheorem{assumption}[theorem]{Assumption}{\bfseries}{\itshape}
\spnewtheorem{fact}[theorem]{Fact}{\bfseries}{\itshape}

\begin{document}

\title{A Distinguisher-Based Attack of a Homomorphic Encryption
  Scheme Relying on Reed-Solomon Codes} 

\author{Val\'erie Gauthier\inst{1}, Ayoub Otmani\inst{1} \and Jean-Pierre Tillich\inst{2}}
\institute{
GREYC - Universit\'e de Caen - Ensicaen\\
Boulevard Mar\'echal Juin, 14050 Caen Cedex, France.\\
\email{valerie.gauthier01@unicaen.fr, ayoub.otmani@unicaen.fr},
\and
SECRET Project - INRIA Rocquencourt \\ 
Domaine de Voluceau, B.P. 105   
78153 Le Chesnay Cedex - France \\
\email{jean-pierre.tillich@inria.fr}
}

\maketitle
\begin{center}
  \date{}
\end{center}

\begin{abstract}
Bogdanov and Lee   suggested a homomorphic public-key encryption scheme based on error correcting codes. 
The underlying public code is a modified Reed-Solomon code obtained
from  inserting a zero submatrix in the Vandermonde generating matrix  defining it. The columns that define 
this submatrix are kept secret and form a set $L$. We give here a distinguisher that detects if one or several columns belong
to $L$ or not. This distinguisher is obtained by considering the  code  generated by component-wise products of codewords of the public code
(the so called ``square code''). This operation is applied to punctured versions of this square code obtained by picking a subset 
 $I$ of the whole set of columns. It turns out that the dimension of
 the punctured square code  is directly related to  the cardinality of
 the intersection of $I$ with $L$. 
This allows an attack which  recovers the full set $L$ 
and  which  can then decrypt any ciphertext. 
 \end{abstract}

%\textbf{Keywords.} Code-based cryptography, homomorphic encryption, distinguisher.
%\section{Introduction}

\section{Introduction}

The concept of homomorphic encryption was first proposed in 1978
in \cite{RAD78}. But it took more than three decades to see the first scheme of this kind
\cite{gentry09}. It is  based on
ideal lattices. Since then several proposals have been made, 
most of them rely on lattice theory. 
One challenging issue is to come up with a homomorphic encryption
scheme using different security assumptions. Recently, the first
symmetric homomorphic encryption scheme based on error-correcting
codes was proposed in \cite{AAPS11}. This work was then followed by
\cite{BL12} which can be considered as the first public-key 
homomorphic scheme based on coding theory. This particular
cryptosystem heavily relies on properties of Reed-Solomon codes.
These codes have been suggested for the first time in a public-key
cryptosystem in \cite{Niederreiter86} but it was shown to be insecure in 
\cite{SidelShesta92}. The attack recovers the underlying Reed-Solomon
allowing the decoding of any encrypted data obtained from a McEliece-type 
cryptosystem based on them. The McEliece cryptosystem
\cite{McEliece78} on the other hand uses Goppa codes.
Since its apparition, it has withstood many attacks and after more than thirty years now, it still belongs 
to the very few unbroken public key cryptosystems. This situation substantiates the claim that inverting 
the encryption function, and in particular recovering the private key from public data, is intractable. 

No significant  breakthrough has been observed with respect to the problem of 
recovering the private key \cite{Gibson91,LS01}. This has led to claim that the generator matrix of a binary Goppa 
code does not disclose any visible structure that an attacker could exploit. This is strengthened by the fact that Goppa codes
share many characteristics with random codes: for instance they asymptotically meet the Gilbert-Varshamov bound, they typically have 
a trivial permutation group, \textit{etc.}  Hence, the hardness of the Goppa code distinguishing problem, which asks whether a Goppa code can be 
distinguished from a random code,  has become 
a classical belief in code-based cryptography, and as a consequence,  a mandatory assumption to prove semantic security
in the random oracle model \cite{DBLP:journals/dcc/NojimaIKM08}, CCA2 security in the standard model \cite{DBLP:conf/ctrsa/DowsleyMN09} and security in the random oracle model 
against existential forgery \cite{CouFinSen01,DBLP:conf/weworc/Dallot07} of the signature scheme \cite{CouFinSen01}.

\medskip

In \cite{FGOPT11a}, an algorithm  that manages to 
distinguish between a random code and a Goppa code has been introduced.
This work without undermining the security of \cite{McEliece78}
prompts to wonder whether it would be possible to devise an attack based 
on such a distinguisher. 
It was found out in \cite{MP12a} that our distinguisher \cite{FGOPT11a} has an equivalent but simpler description
in terms of the component-wise product of codes. This notion was first put forward in coding theory to unify many 
different algebraic decoding algorithms \cite{Pel92,Kot92a}. This distinguisher is even more powerful in the case of Reed-Solomon codes 
than for Goppa codes because,
whereas for Goppa codes it is only successful for rates close to $1$, it can distinguish 
Reed-Solomon codes of any rate from random codes. 
In the specific case of \cite{BL12}, the underlying public code is a modified Reed-Solomon code obtained
from  inserting a zero submatrix in the Vandermonde generating matrix  defining it and in this case our distinguisher leads to an attack.
We present namely in this paper a key-recovery attack on the Bogdanov-Lee homomorphic scheme based on the version of our distinguisher presented in \cite{MP12a}. Our attack runs in polynomial time and is efficient: it only amounts to calculate the ranks of certain matrices derived from the public key.

 More precisely, in their cryptosystem the columns that define 
 the zero submatrix are kept secret and form a set $L$. We give here a distinguisher that detects if one or several columns belong
to $L$ or not. This distinguisher is obtained by considering the code  generated by component-wise products of codewords of the public code
(the so called ``square code"). This operation is applied to punctured versions of this square code obtained by picking a subset 
 $I$ of the whole set of columns. It turns out that the dimension of
 the punctured square code  is directly related to  the cardinality of
 the intersection of $I$ with $L$. 
This allows an attack which  recovers the full set $L$ 
and  which  can then decrypt any ciphertext. 

\medskip

It should also been pointed out that the properties of Reed-Solomon codes with respect to the component-wise product of codes have 
already been used to cryptanalyze a McEliece scheme based on subcodes of Reed-Solomon codes
\cite{Wie10}. The use of this product is nevertheless different in
\cite{Wie10} from the way we use it here.
Note also that our attack is not an 
adaptation of  the Sidelnikov and Shestakov approach \cite{SidelShesta92}. 
Our approach is completely new: it illustrates how a distinguisher
that detects an abnormal behaviour can be used to recover the private
key.

\medskip

In Section~ \ref{sec:basics} we recall important notions from coding
theory. In Section~\ref{Sec:BL} we introduce the cryptosystem and in
Section~\ref{attackBL} we present the key recovery attack.

\section{Reed-Solomon codes and the square code}
\label{sec:basics}
We recall in this section a few relevant results and definitions from coding theory and bring in the
fundamental notion which is used in the attack, namely the square code.
A linear \emph{code} $\code{C}$ of {\em length} $n$ and {\em dimension} $k$ over a finite field $\fq$ of $q$ elements is a subspace of dimension $k$ of the full
space $\fq^n$. 
 It is generally specified by a 
full-rank matrix called a generator matrix which is a $k\times n$ matrix $\Gp$ (with $k \leq n$) over $\fq$ whose
rows span the code:
 $$\code{C} = \left\{\uv \Gp ~|~ \uv \in \fq^k \right\}.$$
 It can also be specified by a {\em parity-check} matrix $\Hm$, which is a matrix whose right kernel is equal to 
 the code, that is
 $$
 \code{C} =  \left\{\xv \in \fq^n ~| \Hm \xv^T=0 \right\},
 $$
 where $\xv^T$ stands for the column vector which is the transpose of the row vector $\xv$.
 The {\em rate} of the code is 
given by the ratio $\frac{k}{n}$.
Code-based public-key cryptography focuses on linear codes that have a polynomial time decoding algorithm. The role of decoding 
algorithms is to correct  errors of prescribed weight. We say that a decoding algorithm
%$\code{C}$ 
corrects $t$ errors if it recovers $\uv$ from the
knowledge of $\uv\Gp + \ev$ for all possible
$\ev \in \F_q^n$ of weight at most $t$. 

%Reed-Solomon codes are very important in coding theory because of 
%the existence of  decoding algorithms. 
Reed-Solomon codes form a special case of codes with a very powerful low complexity decoding algorithm.
It will be convenient to use the definition of Reed-Solomon codes and generalized Reed-Solomon codes as
{\em evaluation codes} 
\begin{definition}[Reed-Solomon code and generalized Reed-Solomon code] \label{defGRS}
Let $k$ and $n$ be integers such that $1 \leqslant k < n \leqslant q$ where $q$ is a power of
a prime number.
Let $\xv = (x_1,\dots{},x_n)$ be an $n$-tuple of distinct elements of
$\fq$. 
The \emph{Reed-Solomon} code  $\RS{k}{\xv}$ of dimension $k$ is the set of 
$(p(x_1),\dots{},p(x_n))$
when $p$ ranges over all polynomials of degree $\leqslant k-1$
with coefficients in $\fq$.
The generalized Reed-Solomon code $\GRS{k}{\xv}{\yv}$ of dimension $k$ is associated to a couple 
$(\xv,\yv) \in \fq^n \times \fq^n$ where $\xv$ is chosen as above and
the entries $y_i$ are arbitrary non zero elements in $\fq$. 
It is defined as the set of  $(y_1p(x_1),\dots{},y_np(x_n))$
where $p$ ranges over all polynomials of degree $\leqslant k-1$
with coefficients in $\fq$.
\end{definition}

Generalized Reed-Solomon codes are quite important in coding theory due to the conjunction of several factors such as :\\
(i) their minimum distance $d$ is maximal among all codes of the same dimension and length since they are MDS codes (their distance is equal to
$n-k+1$),\\
(ii) they can be efficiently decoded in polynomial time when the number of errors
is less than or equal to $\lfloor \frac{d-1}{2} \rfloor = \lfloor \frac{n-k}{2} \rfloor$. 

It has been suggested to use them in a public-key cryptosystem for the first time in \cite{Niederreiter86} but 
it was discovered that this scheme is insecure in \cite{SidelShesta92}. Sidelnikov and Shestakov namely showed that it is 
possible  
to recover in polynomial time for any generalized Reed-Solomon code a possible couple $(\xv,\yv)$ which defines it.
This is all what is needed to decode efficiently such codes and is therefore enough to break the Niederreiter cryptosystem suggested in 
\cite{Niederreiter86} or a McEliece type cryptosystem \cite{McEliece78} when Reed-Solomon are used instead of Goppa codes.

We could not find a way to adapt the Sidelnikov and Shestakov approach for cryptanalyzing the Bogadnov and Lee
cryptosystem. However a Reed-Solomon displays a quite peculiar property with respect to the component-wise
product which is denoted by  $\av \star \bv $ for two vectors 
 $\av=(a_1, \dots, a_n)$ and $\bv=(b_1, \dots, b_n)$ and which is defined by
 $\av \star \bv \eqdef (a_1b_1,\dots{},a_n b_n)$. This can be seen by bringing in the following definition

\begin{definition}[star product of two codes, square code]
Let $\code{A}$ and $\code{B}$ be two codes of length $n$. The
\emph{star product code} denoted by $<\code{A} \star \code{B}>$ of $\code{A}$ and $\code{B}$ is the vector space
spanned by all products $\av \star \bv$ where $\av$ and $\bv$ range over $\code{A}$ and $\code{B}$ respectively.
When $\code{B} = \code{A}$,  $<\code{A} \star \code{A}>$ is called the \emph{square code} of $\code{A}$
and is denoted by $<\code{A}^2>$.
\end{definition}

It is clear that $<\code{A} \star \code{B}>$ is also generated by the $\av_i \star \bv_j$'s where the $\av_i$'s and the
$\bv_j$'s form a basis of $\code{A}$ and $\code{B}$ respectively.
Therefore
\begin{proposition}
 $$\dim(<\code{A} \star \code{B}>) \leq \dim(\code{A}) \dim(\code{B}).$$ 
\end{proposition}
We expect that the square code when applied to a random linear code should be a code of dimension of
order $\min\left\{\binom{k+1}{2},n\right\}$. Actually by using the proof technique of \cite{FGOPT11a} it can be shown for instance that with probability 
going to $1$ as $k$ tends to infinity, the square  
code is of dimension $\min\left\{ \binom{k+1}{2}(1+o(1)),n\right\}$ when $k$ is of the form $k=o(n^{1/2})$.
 On the other hand generalized Reed Solomon codes behave
in a completely different way

\begin{proposition}
$<\GRS{k}{\xv}{\yv}^2>=\GRS{2k-1}{\xv}{\yv\star\yv}$.
\end{proposition}

This follows immediately from the definition of a generalized Reed Solomon code as an evaluation code since
the star product of two elements $\cv=(y_1 p(x_1),\dots,y_np(x_n))$ and $\cv'=(y_1 q(x_1),\dots,y_nq(x_n))$ of $\GRS{k}{\xv}{\yv}$ where
$p$ and $q$ are two polynomials of degree at most $k-1$ is of the form 
$$\cv \star \cv' = (y_1^2 p(x_1)q(x_2),\dots,y_n^2 p(x_n)q(x_n))=(y_1^2r(x_1),\dots,y_n^2r(x_n))$$
where $r$ is a polynomial of degree $\leq 2k-2$. Conversely, any element of the form $(y_1^2r(x_1),\dots,y_n^2r(x_n))$
where $r$ is a polynomial of degree less than or equal to $2k-1$  is a linear combination of star products of two elements of $\GRS{k}{\xv}{\yv}$.

This proposition shows
that the square code is only of dimension $2k-1$ when $2k-1 \leq n$, which is quite unusual.
This property can also be used in the case $2k-1 >n$. To see this, consider the dual of the Reed-Solomon code.
The {\em dual} $\code{C}^\perp$ of a code $\code{C}$ of length $n$ over $\fq$ is defined by 
$$
\code{C}^\perp = \left\{\xv \in \fq^n| (\xv,\yv)=0, \yv \in \code{C}\right\},
$$ 
where $(\xv,\yv)= \sum x_i y_i$ stands for the standard inner product between elements of $\fq^n$.
The dual of a generalized Reed-Solomon code is itself a generalized Reed-Solomon code, see
\cite[Theorem 4, p.304]{MacSloBook}

 \begin{proposition}\label{pr:dual}
 $$
 \GRS{k}{\xv}{\yv}^\perp = \GRS{n-k}{\xv}{\yv'} 
 $$
 where the length of $\GRS{k}{\xv}{\yv}$ is $n$ and $\yv'$ is a certain element of $\fq^n$ depending 
 on $\xv$ and $\yv$.
 \end{proposition}

 Therefore when $2k-1 > n$ a Reed-Solomon code $\GRS{k}{\xv}{\yv}$ can also be distinguished from 
 a random linear code of the same dimension by computing the dimension of
 $<\left(\GRS{k}{\xv}{\yv}^\perp\right)^2>$. We have in this case
 $$<\left(\GRS{k}{\xv}{\yv}^\perp\right)^2>
 =<\GRS{n-k}{\xv}{\yv'}^2>
 = <\GRS{2n-2k-1}{\xv}{\yv'\star\yv'}>
 $$
 and we obtain a code of dimension $2n-2k-1$.

 The star product of two codes is the fundamental notion used in the decoding algorithm based on an error correcting pair
 \cite{Pel92,Kot92a} which unifies common ideas to many algebraic decoding algorithms. It has been used for the first time to cryptanalyze a McEliece scheme based on subcodes of Reed-Solomon codes
 \cite{Wie10}. The use of the star product is nevertheless different in \cite{Wie10} from the way we use it here. In this paper,
 the star product is used to identify  for a certain subcode  $\code{C}$ of a generalized Reed-Solomon code $\GRS{k}{\xv}{\yv}$
 a possible pair $(\xv,\yv)$. This is achieved by computing $<\code{C}^2>$ which in the case which is considered turns out to 
 be equal to $<\GRS{k}{\xv}{\yv}^2>$ which is equal to $\GRS{2k-1}{\xv}{\yv \star \yv}$. The Sidelnikov and Shestakov is then 
 used on  $<\code{C}^2>$ to recover a possible $(\xv,\yv\star\yv)$ pair to describe $<\code{C}^2>$ as a generalized Reed-Solomon
 code. From this, a possible $(\xv,\yv)$ pair for which $\code{C} \subset \GRS{k}{\xv}{\yv}$ is deduced.

%actually the \emph{unique} private key of the scheme.

\section{The Bogdanov-Lee Cryptosystem} \label{Sec:BL}
The cryptosystem proposed by Bogdanov and Lee in \cite{BL12} is a 
public-key homomorphic encryption scheme based on linear codes. 
It encrypts a plaintext $m$ from $\fq$  into a ciphertext $\cv$ that belongs to  
$\fq^n$ where $n$ is a given integer. The key generation requires a non-negative 
integer $\ell$ such that $3\ell < n$ and a subset $L$ of $\{1,\dots{},n\}$ of cardinality $3\ell$.
A set of $n$ distinct elements $x_1,\dots, x_n$ from $\fq$ are 
generated at random.  They serve to construct a $k \times n$ matrix $\Gm$
whose $i$-th column $\Gm^T_i$ ($1\leqslant i \leqslant n$)
is defined by:
\begin{equation*}
\Gv_i^T \eqdef \left\{
\begin{array}{ll}
(x_i, x_i^2, \dots, x_i^{\ell}, 0, \dots , 0 )                  &
\text{if } i \in L\\
& \\
(x_i, x_i^2, \dots, x_i^{\ell}, x_i^{\ell+1}, \dots , x_i^{k} ) & \text{if } i \notin L
\end{array} \right.
\end{equation*}
where the symbol $^T$ stands for the transpose.

In other words, 
when $L$ is the set $\{1,\dots{},3 \ell\}$, $\Gm$ is the following matrix:
\begin{equation*}
\left(
\begin{array}{cccccc}
x_1  & \ldots & x_{3\ell} & x_{3\ell +1} & \ldots & x_{n}\\
\vdots &   & \vdots & \vdots &  & \vdots\\
x_1^{\ell} & \ldots & x_{3\ell}^{\ell} & x_{3\ell +1}^{\ell} & \ldots & x_n^{\ell}\\
0  & \ldots & 0 & x_{3\ell +1}^{\ell+1} & \ldots & x_n^{\ell+1}\\
\vdots &    & \vdots & \vdots &  & \vdots\\
0 &  \ldots & 0 & x_{3\ell +1}^{k} & \ldots & x_n^{k}\\
\end{array} \right).
\end{equation*}

The cryptosystem is now defined as follows.
\begin{itemize}
	\item \textbf{Secret key:} $(L, \Gm)$.
	\item \textbf{Public key:} $\Pm \eqdef \Sm \Gm $ where $\Sm$
          is a $k \times k$ random invertible  matrix over $\fq$.
	\item \textbf{Encryption:} the ciphertext $\cv \in \fq^n$ of a plaintext
          $m \in \fq$ is obtained by  picking $\xv$ in $\fq^k$ uniformly at random and $\ev$
                          in $\fq^n$  by choosing its components according to a certain distribution $\tilde{\eta}$, then computing
			$\displaystyle \cv \eqdef \xv \Pm  + m \textbf{1} +  \ev$ 
               where $\textbf{1} \in \fq^n$ is the all-ones row vector.

	\item \textbf{Decryption:} the
          linear system \eqref{decryptionsystem} is solved for $\yv \eqdef(y_1,\dots{},y_n) \in \fq^n$: 

			%\item \textbf{Input:} $(L, \Gm)$ and $\cv \in \fq^n$.

				  \begin{equation} \label{decryptionsystem}
                   \left\{
                   \begin{array}{rcl}
                   % JP : je n'aime pas du tout.
                   % \displaystyle \sum_{i \in L} y_i \Gv_i    & = & 0\\
                    \displaystyle \Gv \yv^T  & = & 0\\
                    \displaystyle \sum_{i \in L} y_i & = & 1\\
                    y_i&=&0   \textrm{ for all }  i \notin L.
                   \end{array} \right.
                  \end{equation} 
            The plaintext is $m = \displaystyle  \sum_{i = 1}^n y_i c_i$.

\end{itemize}

The decryption algorithm will output  the correct plaintext when 
$\ell$ and $n$ are chosen such that  the entry $e_i$ at position $i$ of the error vector is zero when $i \in L$.
The distribution $\eta$ which is used to draw at random the coordinates of $\ev$ is chosen such that this
property holds with very large probability.
To check the correctness of the algorithm when this property on $\ev$ holds, notice that  the linear system~\eqref{decryptionsystem} has $3\ell$ unknowns and $\ell +1$ equations and since it is by construction of rank $\ell+1$, it always admits at least one solution. Then observe that
\begin{eqnarray*}
\sum_{i = 1}^n y_i c_i &=& (\xv \Pm + m \textbf{1} + \ev) \yv^T\\
& =&  (\xv \Pm + m \textbf{1} ) \yv^T \;\;\;\text{   (since $e_i=0$ if $i\in L$ and $y_i=0$ if $i \notin L$)}\\
& = &\xv \Sm \Gm \yv^T + m \sum_{i=1}^n y_i\\
& = & m \;\;\;\text{   (since $\Gv \yv^T=0$ and  $\sum_{i=1}^n y_i=1$)}.
\end{eqnarray*}

The parameters $k,q,\ell$ and the noise distribution $\tilde{\eta}$ are chosen such as
\begin{itemize}
\item $q = \Omega\left(2^{n^\alpha}\right)$;
\item $k=\Theta\left(n^{1-\alpha/8}\right)$;
\item the noise distribution $\tilde{\eta}$ is the $q$-ary symmetric channel with noise rate $\eta = \Theta\left(1/n^{1-\alpha/4}\right)$,
that is $\prob(e_i=0)=1-\eta$ and $\prob(e_i=x)=\frac{\eta}{q-1}$ for any $x$ in $\fq$ different from zero;
\item $\ell = \Theta\left(n^{\alpha/4}\right);$
\end{itemize}
where $\alpha$ is some constant in the range $(0,\frac{1}{4}]$. It is readily checked that the probability that $e_i \neq 0$ for $i \in L$ is 
vanishing as $n$ goes to infinity since it is upper-bounded by $\eta \ell = \Theta\left( \frac{n^{\alpha/4}}{n^{1-\alpha/4}}\right)=\Theta\left( n^{-1+\alpha/2}\right)=o(1)$.

\section{An efficient attack on the Bogdanov-Lee homomophic cryptosystem}\label{attackBL}

\subsection{Outline}
The attack consists in first recovering the secret set $L$ and from here finds directly a suitable vector $\yv$ by solving the system
\begin{equation} \label{decryptionsystempub}
      \left\{ 
      \begin{array}{lcl}
      \displaystyle  \Pv \yv^T    & = & 0\\
      \displaystyle \sum_{i \in L} y_i & = & 1\\
       y_i&=&0   \textrm{ for all }  i \notin L.
      \end{array} \right.
\end{equation}
Indeed, requiring that $\Pv \yv^T     =  0$  is equivalent to $\Sv \Gv \yv^T=0$ and since $\Sv$ is invertible this is equivalent to the
equation $\Gv\yv^T=0$. Therefore System \eqref{decryptionsystempub} is equivalent to the
``secret'' system  \eqref{decryptionsystem}. An attacker may therefore recover $m$ without even knowing $\Gv$ just 
by outputting $\sum_i y_i c_i$ for any solution $\yv$ of \eqref{decryptionsystempub}.
In the following subsection, we will explain how   $L$ can be recovered from $\Pv$ in 
polynomial time. 

\subsection{Recovering $L$}\label{attack}

Our attack relies heavily on the fact that the public matrix may be viewed as a the generator matrix of a code 
$\code{C}$ which 
is quite close to a generalized Reed-Solomon code (or to a Reed-Solomon if a row consisting only of $1$'s is added
to it). Notice that any punctured version of the code has also this property (a punctured code consists in 
keeping only a fixed subset of positions in a codeword). More precisely, let us introduce 
\begin{definition}
For any $I \subset \{1,\dots{},n\}$ of cardinality $\card{I}$, the restriction of a code $\code{A}$ of length $n$ is the subset of $\fq^{|I|}$ defined as:
$$
\code{A}_I \eqdef \Big \{  \vv  \in \fq^{\card{I}} \mid \exists \av \in \code{A}, \vv = (a_i)_{i \in I} \Big \}.
$$  
\end{definition}

The results about the unusual dimension of the square of a Reed-Solomon codes which are given in 
Section \ref{sec:basics} prompt us to study the dimension of the square code $<\code{C}^2>$ or more generally the
dimension of $<\code{C}_I^2>$. When $I$ contains no positions in $L$, then $\code{C}_I$ is nothing but a generalized Reed-Solomon 
code and we expect a dimension of $2k-1$ when $|I|$ is larger than $2k-1$. On the other hand, when there are positions in $I$
which also belong to $L$ we expect the dimension to become bigger and  the dimension of $<\code{C}^2>$ to behave 
as an increasing function of $|I \cap L|$. This is exactly what happens as shown in the proposition below.
 
\begin{proposition} \label{prop:dsg}
Let $I$ be a subset of $\{1,\dots{},n\}$ and set $J
\eqdef I \cap L$. 
If the cardinality of $I$ and
$J$ satisfy $\card{J} \leqslant
\ell-1$ and $\card{I} -  \card{J} \geqslant 2k$ then
\begin{equation}
\dim(<\code{C}_I ^2>) = 2k - 1 + \card{J}.
\end{equation}
\end{proposition}

The proof of this proposition can be found  in Appendix \ref{sec:proof_prop:dsg}.
An attacker can exploit this proposition to mount a distinguisher that recognizes whether a given position belongs 
to the secret set $L$. At  first  a set $I$ which satisfies with high probability the assumptions of
Proposition \ref{prop:dsg} is randomly chosen. Take for instance $|I|=3k$. Then $d_I \eqdef \dim(<\code{C}_I^2> )$ is computed. Next,  one element $x$ is removed from $I$ to get a new set
$I'$ and  $d_{I'} = \dim(<\code{C}_{I'}^2> )$ is computed. The only two possible cases are then:
\begin{enumerate}
  \item if $x \notin L$ then $d_{I'} = d_{I}$ 
	\item and if $x \in L$ then $d_{I'} = d_{I}-1$.
\end{enumerate}
By repeating this procedure, the whole set $J=I \cap L$ is easily recovered. The
next step now is to find all the elements of $L$ that are not in
$I$. One solution is to exchange one element in $I \setminus J$ by another element in
$\{1,\dots{},n\} \setminus I$ and compare the values of $d_I$. If it
increases, it means that the new element belongs to $L$. At the
end of this procedure the set $L$ is totally recovered. This probabilistic algorithm is obviously of  polynomial time complexity
and breaks completely the homomorphic scheme suggested in \cite{BL12}.

\bibliographystyle{alpha}
\bibliography{crypto%
%,codecrypto
}

\appendix

\newpage

\section{Proof of Proposition~\ref{prop:dsg}}
\label{sec:proof_prop:dsg}

The proof of Proposition~\ref{prop:dsg} proceeds by exhibiting a basis 
 of the linear space $<\code{C}_I^2>$.  For this purpose  we define for any $t$ in $\{1,\dots{},k\}$, 
$\Xv^t \eqdef
(x^t_i)_{i \in I}$  and $\Yv^t \eqdef (Y^t_i)_{i \in I}$ 
with:
\begin{equation*}
Y_i^t \eqdef \left\{
\begin{array}{ccl}
0       & & \text{if } i \in J\\
        & &\\
x_i^t   & & \text{if } i \in I \setminus J.
\end{array} \right.
\end{equation*}
%JP : cette notation n'est jamais utilise dans la suite, je la vire.
%\begin{equation*}
%y_i^t \eqdef \left\{
%\begin{array}{ccl}
%0       & & \text{if } i \in L\\
%        & &\\
%x_i^t   & & \text{if } i \in I \setminus L.
%\end{array} \right.
%\end{equation*}
Notice that  $\code{C}_I$ is  the  vector space spanned by the  $\Xv^t$'s
for $1 \leqslant t \leqslant \ell$ and the $\Yv^t$'s for $\ell +1
\leqslant t \leqslant k$. 

The proof of Proposition \ref{prop:dsg} starts by giving a generating set for 
$<\code{C}_I^2>$.

\begin{lemma}\label{lem:generating_set}
$<\code{C}_I^2>$ is generated by the set of vectors $\Xv^t$ for
$2 \leqslant t \leqslant 2\ell$ and $\Yv^t$ for 
$\ell + 2 \leqslant t \leqslant 2k$.
\end{lemma}

\begin{proof}
Let us define: 
\begin{equation*}
\Zv^t \eqdef \left\{
\begin{array}{ll}
\Xv^t       & \text{if } 1 \leqslant t \leqslant \ell \\
&\\
\Yv^t & \text{if } \ell +1 \leqslant t \leqslant k.
\end{array} \right.
\end{equation*}
Obviously, $<\code{C}_I^2>$ is generated by
the vectors $\Zv^r \star \Zv^s$  where $r$ and $s$ range over $\{1,\dots{},k\}$.
We notice now that
\begin{equation*}
\Zv^r \star \Zv^s \eqdef \left\{
\begin{array}{ll}
\Xv^{r+s}       & \text{if } r \text{ and }s \in \{1,\dots{},\ell \}\\
               & \\
\Yv^{r+s}       & \text{if } r \text{ or } s \notin \{1,\dots{},\ell \}.
\end{array} \right. 
\end{equation*}
In particular, the following equality holds:
$$
\Big \{ \Zv^r \star \Zv^s ~|~ 1 \leqslant r \leqslant k \text{ and } 1 \leqslant s \leqslant k\Big \} = 
\Big \{ \Xv^t ~|~ 2 \leqslant t \leqslant 2\ell \Big \} \bigcup \Big \{ \Yv^t ~|~ \ell+2 \leqslant t \leqslant 2k \Big \}.
$$ \qed
\end{proof}

The next step is to find some linear relations between the $\Xv^t$'s and the $\Yv^t$'s.
This is achieved by

\begin{lemma} \label{lem:dep}
If    $\ell + \card{J} + 2 \leqslant t \leqslant   2\ell$,
then $\Xv^t$ belongs to the vector space generated by
  $$\displaystyle \bigcup_{ u = \ell+2}^{t-1} \Big\{\Xv^u, \Yv^u  \Big\} \bigcup \Big\{\Yv^t \Big\}.$$ 
  \end{lemma}

\begin{proof}
We consider $U$ as an indeterminate and we define the polynomials $\varphi(U)$ and $R(U)$ as
$\displaystyle \varphi(U) \eqdef \prod_{i \in J} (U - x_i)$ and $R(U) \eqdef \varphi(U) U^{t-\card{J}}$. The degree of $R(U)$ 
is equal to $t$ and hence satisfies $\deg(R) \leqslant 2\ell$. $R(U)$ can also be viewed as the polynomial $\displaystyle \sum_{s = t-\card{J}}^t r_s U^s$
where each $r_s$ belongs to $\fq$ and $r_t = 1$. One can see that by construction of $R(U)$ when $i \in J$ then 
$\displaystyle R(x_i) = \sum_{s = t-\card{J}}^t r_s x_i^s =0$. So if we denote by $\Xv^s_i$ (\textit{resp.} $\Yv^s_i$) 
the entry of $\Xv^s$ (\textit{resp.} $\Yv^s$) at position $i$ we equivalently have when $i \in J$: 
$$
\sum_{s = t-\card{J}}^t r_s \Xv_i^s =  \sum_{s = t-\card{J}}^t r_s x_i^s = R(x_i)= 0. 
$$
By the very definition of $\Yv_i^s$ which is equal to $0$ when $i \in J$, we have that
$\sum_{s = t-\card{J}}^t r_s \Yv_i^s=0$.  On the other hand by definition of $\Xv^s$ and $\Yv^s$, we also have that
$$
\sum_{s = t-\card{J}}^t r_s \Xv_i^s = \sum_{s = t-\card{J}}^t r_s \Yv_i^s
$$
for $i$ in $I \setminus J$. Therefore in all cases we have
$$
\sum_{s = t-\card{J}}^t r_s \Xv_i^s = \sum_{s = t-\card{J}}^t r_s \Yv_i^s,
$$
 and since $r_t = 1$ we 
can write that:
$$
\Xv^t = \sum_{s = t-\card{J}}^t r_s \Yv^s - \sum_{s = t-\card{J}}^{t-1} r_s \Xv^s.
$$
This concludes the proof of the lemma by noticing that $t \geq |J|+\ell+2$ implies that 
the $s$ which appears in the sum above is larger than or equal to  $\ell+2$. \qed
\end{proof}

It remains to prove that the generating set obtained by removing the linear relations 
obtained in Lemma \ref{lem:dep} is now an independent set.  
\begin{proposition} \label{prop:basis}
Assume that $|J| \leq \ell-1$ and $|I|-|J| \geq 2k$, then 
the set of $\Xv^t$'s with $2 \leqslant t \leqslant \ell + \card{J} + 1$ and $\Yv^t$'s with 
$\ell+ 2 \leqslant t \leqslant 2k$ form a basis of $<\code{C}_I^2>$.
\end{proposition}

\begin{proof}
A consequence of Lemma~\ref{lem:dep} is that $\Xv^t$ with 
$2 \leqslant t \leqslant \ell + \card{J} + 1$ and $\Yv^t$ with 
$\ell+ 2 \leqslant t \leqslant 2k$ generate the code $<\code{C}_I^2>$ but it remains to prove that they are
linearly independent. For this purpose, let us assume that there exists a linear relation 
between them \textit{i.e.}, there exist $a_s$ and $b_s$ in $\fq$ for 
$ 2 \leqslant s \leqslant 2k$ such that:
\begin{equation} \label{eq:rellin}
\sum_{s = 2}^{\ell+|J|+1} a_s \Xv^s + \sum_{s = \ell+2}^{2k} b_s \Yv^s = 0.
\end{equation}
By setting $a_s=0$ for $\ell+|J| +2 \leqslant s \leqslant 2k$ and
$b_s=0$ for $2 \leqslant s \leqslant \ell+1$,
Equation~\eqref{eq:rellin} can be rewritten as: 
\begin{equation} \label{eq:rellin2}
\sum_{s = 2}^{2k} \left(a_s \Xv^s + b_s \Yv^s \right) =0.
\end{equation}
Let us denote $\displaystyle R(U) \eqdef \sum_{s = 2}^{2k} (a_s + b_s)
U^s$. We know that if $i \notin J$ then $\Yv_i^s = \Xv_i^s = x_i^s$
for $s$ in $\{2,\dots{},2k\}$. 
Therefore by Equation~\eqref{eq:rellin2} we
have $R(x_i)=0$ for any $i \notin J$. 
As we have assumed that $\card{I} - \card{J} \geqslant 2k$, it implies
that $R(U) = 0$ or equivalently $a_s = - b_s$ for all $s$. In
particular $a_s = 0$ for $2 \leqslant s \leqslant \ell + 1$.
On the other hand, when $i \in J$, we have $\Yv^s_i = 0$ and $\Xv_i^s = x_i^s$ for any
$s$ in $\{\ell+2,\dots{},\ell+\card{J}+1\}$. Hence when $i \in J$, Equation~\eqref{eq:rellin} leads in  fact
to:
\begin{equation} \label{equ:xv}
\sum_{s=\ell+2}^{\ell+\card{J}+1}a_s \Xv_i^s = 0.
\end{equation}
Now let us consider  $\displaystyle Q(U) \eqdef
\sum_{s=\ell+2}^{\ell+\card{J}+1}a_s U^s$ and observe that there exists 
some polynomial $S(U)$ with $\deg(S) \leqslant \card{J} - 1$ such that:
$$
Q(U) = U^{\ell+2} S(U). 
$$
From Equation~\eqref{equ:xv} we know that $Q(x_i) = 0$ for all $i$ in $J$.
Since all $x_i$'s are different from  $0$, this implies that $S(x_i) = 0$. Since $\deg(S) \leqslant
\card{J} - 1$ this means that $S(U) = 0$, and 
therefore $a_s = 0$ for all $ \ell + 2 \leqslant s \leqslant \ell +
\card{J} + 1$.
Then equation~\eqref{eq:rellin} holds if and only if all the
coefficients $a_s$ and $b_s$ are zero, which means that 
$\Xv^t$ with 
$2 \leqslant t \leqslant \ell + \card{J} + 1$ and $\Yv^t$ with 
$\ell+ 2 \leqslant t \leqslant 2k$
form indeed a basis of $<\code{C}_I^2> $ whose dimension
is therefore $2k-1+\card{J}$.
\qed
\end{proof}

Proposition~\ref{prop:dsg} immediately follows from 
 Proposition~\ref{prop:basis} which characterises a basis of $<\code{C}_I^2> $.
\end{document}